\documentclass[11pt,a4paper,reqno]{amsart}
\usepackage[foot]{amsaddr} 
\newif\ifarxiv\arxivtrue
\arxivfalse

\usepackage{amsmath}
\usepackage{amssymb}
\usepackage{hyperref}
\usepackage{color}

\usepackage{xcolor}
\hypersetup{
	colorlinks,
	linkcolor={black!30!blue},
	citecolor={black!30!blue},
	urlcolor={black!30!blue}
}

\usepackage[T1]{fontenc}
\usepackage[utf8]{inputenc}

\usepackage{amsthm}
\usepackage{bbm}

\newtheorem{thm}{Theorem}[section]

\newtheorem{lem}[thm]{Lemma}
\newtheorem{prop}[thm]{Proposition}
\newtheorem{cor}[thm]{Corollary}
\newtheorem{remark}[thm]{Remark}

\def\norm #1{\Vert #1\Vert}

\def\mod{{\mathop{\rm mod}\nolimits}}
\def\ind{{\mathop{\rm ind}\nolimits}\,}

\def\ket #1{\vert#1\rangle}

\def\abs#1{\vert#1\vert}

\def\Vtw{{\widetilde V}}

\def\embf{\textbf}

\newcommand{\Span}{\mathrm{span}}

\newcommand{\Ad}{\mathrm{Ad}}
\newcommand{\Ran}{\mathrm{Ran}}

\newcommand{\triv}{\mathrm{triv}}
\newcommand{\rest}{\mathrm{rest}}

\newcommand{\ed}{{\mathrm e}}
\newcommand{\ac}{\mathrm{ac}}
\newcommand{\iu}{{\mathrm i}}

\newcommand{\ie}{i.e.\;}

\newcommand{\N}{{\mathbb N}} 
\newcommand{\Z}{\mathbb{Z}} 
\newcommand{\R}{{\mathbb R}} 
\newcommand{\C}{{\mathbb C}} 

\DeclareMathOperator{\I}{\mathbbm{1}} 

\newcommand{\caH}{{\mathcal H}}
\newcommand{\caO}{{\mathcal O}}

\usepackage[margin=2.4cm]{geometry}

\makeatletter
\def\subsection{\@startsection{subsection}{2}%
	\z@{.5\linespacing\@plus.7\linespacing}{.5\linespacing}%
	{\normalfont\scshape\centering}}
\def\subsubsection{\@startsection{subsubsection}{2}%
	\z@{.5\linespacing\@plus.7\linespacing}{.5\linespacing}%
	{\normalfont\scshape\centering}}
\makeatother

\numberwithin{equation}{section}
\begin{document}

\title[AC edge spectrum of topological insulators with odd time-reversal symmetry]{Absolutely continuous edge spectrum of topological insulators with an odd time-reversal symmetry}

\nopagebreak

\author{Alex Bols${}^1$}
\email{\href{mailto:alex-b@math.ku.dk}{alex-b@math.ku.dk}}
\address{${}^1${QMATH}, Department of Mathematical Sciences, University of Copenhagen, Universitetsparken 5, 2100 Copenhagen, Denmark}
\author{Christopher Cedzich${}^2$}
\email{\href{mailto:cedzich@hhu.de}{cedzich@hhu.de}}
\address{${}^2$Quantum Technology Group, Heinrich Heine Universit\"at D\"usseldorf, Universit\"atsstr. 1, 40225 D\"usseldorf, Germany}

\begin{abstract}
	We show that non-trivial two-dimensional topological insulators protected by an odd time-reversal symmetry have absolutely continuous edge spectrum. The proof employs a time-reversal symmetric version of the Wold decomposition that singles out ballistic edge modes of the topological insulator.
\end{abstract}

\maketitle


\section{Introduction}

Hall insulators support ballistic chiral edge modes \cite{PhysRevB.25.2185}. In a free electron description, these edge modes are associated to absolutely continuous spectrum filling the bulk gap. The presence of absolutely continuous spectrum has been proven for the Landau Hamiltonian with weak disorder and a steep edge potential or appropriate half-plane boundary conditions using Mourre estimates \cite{MacrisMartinPule1999, FrohlichGrafWalcher2000, DeBievrePule2002, HislopSoccorsi2008, BrietHislopRaikovSoccorsi2009}, and recently for Hall insulators on the lattice using index theory \cite{BolsWerner2021}.

The question naturally arises whether such ballistic modes are also present in topological insulators that are protected by an odd time-reversal symmetry, and for which the Hall conductance vanishes. Because of the time-reversal symmetry any left moving edge mode has a companion right moving edge mode so Mourre estimates, which apply only when the edge modes are strictly chiral, cannot be used to answer this question. In this note we use index theory to show that absolutely continuous edge spectrum is a consequence of a non-trivial $\Z_2$-valued bulk index.

We appeal to the bulk-edge correspondence for time-reversal invariant topological insulators \cite{graf2013bulk, ShapiroEtAl2020, alldridge2020bulk, bols2021fredholm} which links the bulk index to an edge index associated to a time-reversal symmetric unitary acting on the edge modes. Inspired by \cite{AschBourgetJoye2020}, we prove a symmetric Wold decomposition for such unitaries which implies in particular that the absolutely continuous spectrum of this unitary covers the whole unit circle if the edge index is non-trivial. The Hamiltonian describing the system with edge is then shown to inherit this absolutely continuous spectrum.

\section{Setup and Results}

\subsection{Edge spectrum of time-reversal symmetric topological insulators}

We consider free electrons moving on the lattice $\Z^2$ modeled by a bulk Hamiltonian $H$ on $\ell^2(\Z^2, \C^n)$ that is exponentially local in the sense that for all $\vec x, \vec y \in \Z^2$
\begin{equation}
\norm{ P_{\vec x} H P_{\vec y}} \leq C \ed^{-\norm{\vec x - \vec y} / \xi}
\end{equation}
for some $C < \infty$, $\xi > 0$, where $P_{\vec x}$ denotes the projection onto the site at $\vec x$.

Moreover, we take $H$ to be invariant under an \embf{odd time-reversal symmetry}, \ie there is an anti-unitary operator $\tau$ with $\tau^2 = -\I$ such that $\tau H \tau^* = H$. The time-reversal symmetry is further assumed to act ``on-site'' meaning that $\tau X_i \tau^* = X_i$ for $i = 1, 2$ where $X_i$ denotes the position operator in the $i$-direction. We further assume that $H$ has a \embf{bulk gap}, \ie that for some open interval $\Delta\subset\R$
\begin{equation*}
\Delta \cap \sigma(H) = \emptyset.
\end{equation*}
For any $\mu \in \Delta$ we denote the Fermi projection by $P_F = \chi_{\leq \mu}(H)$. 

We define a unitary that models the insertion of a unit of magnetic flux at the origin:
\begin{equation} \label{eq:flux insertion}
U := \ed^{\iu \arg \vec X}
\end{equation}
with $\vec X = (X_1, X_2)$ the vector of position operators and $\arg (x_1, x_2) = \arctan(x_2 / x_1)$. 
It follows from the discussion around Lemma 1 in \cite{AizenmanGraf1998} that the difference $A_B := U P_F U^* - P_F$ is compact, so we can define a \embf{bulk index}
\begin{equation}
\ind_2^{B}(H, \Delta) := \dim \ker (A_B - \I) \,\,\, \mod \, 2.
\end{equation}
This bulk index for odd time-reversal invariant topological insulators was first defined in \cite{KomaKatsura2016}. It is a non-commutative extension of the Kane-Mele invariant \cite{KaneMele2005} and related to the index of a pair of projections introduced in \cite{AvronSeilerSimon1994}.

With the bulk Hamiltonian $H$ we associate a half-space Hamiltonian $\hat H$ on the half-space $\ell^2(\Z \times \N, \C^n)$ that is also exponentially local and agrees with the bulk Hamiltonian in the bulk in the sense that
\begin{equation}
\norm{ P_{\vec x} (\iota^{\dag} H \iota - \hat H) P_{\vec y}} \leq C \ed^{- \norm{\vec x - \vec y}/\xi   - y_2 / \xi' }
\end{equation}
for some $C < \infty, \xi, \xi' > 0$ and all $\vec x, \vec y \in \Z \times \N$. Here $\iota:\ell^2(\Z \times \N, \C^n) \to\ell^2(\Z \times \Z, \C^n)$ denotes the injection that is induced by the natural inclusion of the half-space lattice $\Z \times \N$ in the bulk lattice $\Z \times \Z$. 
Since the time-reversal symmetry $\tau$ acts on-site, it naturally restricts to the half-space and we denote this restriction also by $\tau$. We assume that the half-space Hamiltonian is also time-reversal invariant, \ie $\tau \hat H \tau^* = \hat H$. 

The main result of the paper is the following:
\begin{thm} \label{thm:edge spectrum}
	If $\ind_{2}^B(H, \Delta) = 1$ then
	\begin{equation}
	\Delta \subset \sigma_{ac}(\hat H),
	\end{equation}
	\ie the half-space Hamiltonian $\hat H$ has absolutely continuous spectrum everywhere in the bulk gap.
\end{thm}

\begin{remark}
	Together with \cite{BolsWerner2021} this covers all non-trivial cases of two-dimensional free-fermion topological insulators in the periodic table of topological insulators \cite{kitaevPeriodic}.
\end{remark}

\subsection{Time-reversal symmetric Wold decomposition} \label{sec:symmetric Wold statement}

The proof of Theorem \ref{thm:edge spectrum} builds on the following time-reversal symmetric version of Theorem 2.1 of \cite{AschBourgetJoye2020} which is of independent interest:
\begin{thm} \label{thm:symmetric Wold}
	Let $U$ be a unitary and $P$ a projection such that $A = U P U^* - P$ is compact. Moreover, assume that $\tau U \tau^* = U^*$ and $\tau P \tau^* = P$ for an odd time-reversal symmetry $\tau$. Then there exists a unitary $W$ with $U - W$ compact and $\tau W \tau^* = W^*$ such that:
	\begin{itemize}
		\item If $\dim(A - \I) \,\,\, \mod \,2 = 0$, then $[W, P] = 0$.
		\item If $\dim(A - \I) \,\,\, \mod \,2 = 1$, then $W = S \oplus W_\triv$ where $[W_\triv, P] = 0$ and the unitary $S$ which we define in Eq. \eqref{eq:shiftau} consists of two opposite shift operators and therefore has absolutely continuous spectrum covering the whole unit circle.
	\end{itemize}
	Moreover, if $A$ is Schatten-$p$, then so is $U - W$.
\end{thm}
Since the absolutely continuous spectrum of an operator is stable under trace-class perturbations, Theorem \ref{thm:symmetric Wold} immediately implies
\begin{cor}\label{cor:absolutely continuous spectrum}
	Let $U$ be a unitary and $P$ a projection such that $A = U P U^* - P$ is trace-class. Moreover, assume that $\tau U \tau^* = U^*$ and $\tau P \tau^* = P$ for an odd time-reversal symmetry $\tau$. Then, if $\dim \ker (A - \I) \,\,\, \mod \, 2 = 1$ the absolutely continuous spectrum of $U$ covers the whole unit circle.
\end{cor}

The unitary $S$ appearing in the statement of Theorem \ref{thm:symmetric Wold} is defined as follows: Consider the Hilbert space $\ell^2(\Z, \C^2) = \ell^2(\Z) \otimes \C^2$ with orthonormal basis $\{\ket{x,\pm}\}_{x \in \Z}$ labeled by position and $\sigma_z$-eigenstates. On this Hilbert space, consider the odd time-reversal symmetry
\begin{equation}
	\tau = \bigoplus_{x \in \Z} \begin{bmatrix} 0 && -1 \\ 1 && 0 \end{bmatrix} K,
\end{equation}
where $K$ is complex conjugation with respect to the basis $\{\ket{x,\pm}\}_{x \in \Z}$. Then $S$ acts as the right shift on the spin-up sector and the left shift on the spin-down sector, \ie
\begin{equation}\label{eq:shiftau}
	S\ket{x,\pm}=\ket{x\pm1,\pm}
\end{equation}
for all $x \in \Z$. Since $S$ contains two copies of the shift operator, it has absolutely continuous spectrum covering the whole unit circle. Moreover, it is straightforward to verify that the $\tau$ above is an odd time-reversal symmetry for $S$, i.e. that  $\tau S \tau^* = S^*$.

\section{Proof of absolute continuity of edge spectrum}

\subsection{Edge index and bulk-edge correspondence}

Let $g : \R \rightarrow [0, 1]$ be a smooth non-increasing function interpolating from $1$ to $0$ such that its derivative is supported in the bulk gap $\Delta$. 
We have
\begin{equation}
	P_F = g(H).
\end{equation}
Consider now the edge unitary $U_E := W_g(\hat H)$ where $W_g$ is the function
\begin{equation} \label{def:W_g}
	W_g : \R \rightarrow \C : x \mapsto \ed^{2 \pi \iu g(x) }.
\end{equation}
This unitary is local and supported near the edge of the half-space. We denote by $\hat \Pi_1$ the projection on the upper right quadrant $\{ \vec x \in \Z \times \N \, | \, x_1 \geq 0 \}$, then
\begin{lem}[\cite{ElbauGraf2002}] \label{lem:W_g_locality}
	The commutator $[U_E, \hat \Pi_1 ]$ is trace class.
\end{lem}

This lemma follows immediately from Lemmas A.2. and A.3. in \cite{ElbauGraf2002}. It follows that $A_E := U_E \hat \Pi_1 U_E^* - \hat \Pi_1 = [U_E, \hat \Pi_1] U_E^*$ is also trace class and in particular compact so the index
\begin{equation}
	\ind_2^{E}(H, \Delta) := \dim \ker (A_E - \I) \,\,\, \mod \, 2
\end{equation}
is well defined. We call this the \embf{edge index}.

By bulk-boundary correspondence, the bulk and edge indices are equal:

\begin{thm}[Theorem 2.11 of \cite{ShapiroEtAl2020}] \label{thm:bulk-edge}
	Under the above assumptions on $H$ we have
	\begin{equation}
		\ind_2^B(H, \Delta) = \ind_2^E(H, \Delta).
	\end{equation}
\end{thm}

\begin{remark}\mbox{}
	\begin{itemize}
		\item The edge index may a priori depend on the boundary conditions defining the half-space Hamiltonian $\hat H$. The bulk-edge correspondence, Theorem \ref{thm:bulk-edge}, implies that this is not the case, justifying our notation $\ind_2^E(H, \Delta)$.
		\item In \cite{ShapiroEtAl2020}, the bulk and edge indices are given as (odd) Fredholm indices \cite{AtiyahSinger1969}, \cite{Schulz2015}. If $[U, P]$ is compact then $F = P U P + P^\perp$ is a Fredholm operator and its odd Fredholm index is $\dim \ker F \,\,\, \mod \, 2$. The equivalence to our $\dim \ker (A - \I) \,\,\, \mod \, 2$ with $A = U P U^* - P$ is easily established by noting that $\psi \in \ker F$ if and only if $U \psi \in \ker(A - \I)$.
	\end{itemize}
\end{remark}

\subsection{Proof of theorem \ref{thm:edge spectrum}}

The following is a verbatim copy of the corresponding proof in \cite{BolsWerner2021}, except for the remark that now $\sigma_{ac}(U_E) = U(1)$ follows, through Corollary \ref{cor:absolutely continuous spectrum}, from an odd time-reversal symmetry and a non-trivial $\Z_2$-valued edge index.

By Theorem \ref{thm:bulk-edge}, $\ind_2^B(H, \Delta) =1$ implies $\ind_2^E(H, \Delta) = \dim \ker (U_E \hat \Pi_1 U_E^* - \hat \Pi_1 - \I) \,\,\, \mod \, 2 
= 1$. Since $\tau$ acts locally, $\tau \hat \Pi_1 \tau^{*} = \hat \Pi_1$, and from $\tau \hat H \tau^* = \hat H$ we get $\tau U_E \tau^* = \tau W_g(\hat H) \tau^* = \overline W_g(\hat H) = U_E^*$. Moreover, $U_E \hat \Pi_1 U_E^* - \hat \Pi_1$ is trace-class by Lemma \ref{lem:W_g_locality}. It therefore follows from Corollary \ref{cor:absolutely continuous spectrum} that the absolutely continuous spectrum of the edge unitary $U_E$ is the whole unit circle.
	
We can choose $g$ in such a way that $x \mapsto W_g(x) := \ed^{2 \pi \iu g(x)}$ is a smooth function that satisfies $\Delta = W_g^{-1} \big(  U(1) {\setminus} \{1\}   \big)$ and such that $W_g$ is invertible on $\Delta$.

Now, let $P_{\Delta} = \chi_{\Delta}(\hat H)$ be the spectral projection of $\hat H$ on the interval $\Delta$. Since the absolutely continuous spectrum of $W_g(\hat H)$ covers the whole unit circle and $W_g$ differs from 1 only on $\Delta$, there is an absolutely continuous spectral measure $\mu_{\ac}$ of $W_g( P_{\Delta} \hat H P_{\Delta} )$ that is supported on the whole unit circle. By spectral mapping (Proposition 8.12 of \cite{Conway2019}), we have that $\mu_{\ac} \circ W_g$ is a spectral measure for $(W_g)|_{\Delta}^{-1} ( W_g( P_{\Delta} \hat H P_{\Delta} ) ) = P_{\Delta} \hat H P_{\Delta}$. The function $W_g$ is smooth and maps the interval $\Delta$ into the unit circle, so we see that $\mu_{\ac} \circ W_g$ is an absolutely continuous measure supported on the entire closed interval $\overline \Delta$. (Supports of measures are closed sets.) This means that $\sigma_{\ac}(P_{\Delta} \hat H P_{\Delta}) = \overline \Delta$, and hence, $\Delta \subset \sigma_{\ac}(\hat H)$ as required.
\hfill\qedsymbol

\section{Proof of the symmetric Wold decomposition}

In this section we prove Theorem \ref{thm:symmetric Wold}. 
We fix a unitary $U$ and a projection $P$ such that $A = U P U^* - P$ is compact. Moreover, we require $\tau U \tau^* = U^*$ and $\tau P \tau^* = P$ for some odd time-reversal symmetry $\tau$. We further write $Q = U P U^*$, and for $I \subset \R$ we denote by $E_I$ the range of the spectral projection $\chi_{I}(A)$ of $A$. We also write $E_{\lambda} = E_{\{\lambda\}}$ for the $\lambda$-eigenspace of $A$.

\subsection{Spectral symmetry and Kramers degeneracy}

Following \cite{AvronSeilerSimon1994} we introduce the operator $B = \I - P - Q$. One easily checks that
\begin{equation}
	A^2 + B^2 = \I, \quad \text{and} \quad AB + BA = 0.
\end{equation}

\begin{lem} \label{lem:spectral symmetry}
	For $\lambda \not\in \{-1, 0, +1\}$ the operator $B$ maps $E_{\lambda}$ isomorphically to $E_{-\lambda}$. In particular, $\dim E_{\lambda} = \dim E_{-\lambda}$.
\end{lem}

\begin{proof}
	Let $A \phi = \lambda \phi$ for $\lambda \not\in \{-1, 0, +1\}$. Since $AB = - BA$ we have
	\begin{equation}
		AB \phi = - BA \phi = -\lambda \phi
	\end{equation}
	and since $B^2 = \I - A^2$ we have that $B^2 \phi = (1 - \lambda^2) \phi \neq 0$, because $\lambda \not\in \{-1, +1\}$. Thus $B \phi \neq 0$, and $B$ maps  $E_{\lambda}$ injectively into $E_{-\lambda}$ and vice versa. Since $A$ is compact, $E_{\lambda}$ is finite dimensional and it follows that $B$ maps $E_{\lambda}$ isomorphically to $E_{-\lambda}$
\end{proof}

Define $\tilde \tau = U \tau$ which satisfies $\tilde \tau^2 = U \tau U \tau = U U^* \tau^2 = - \I$. From straightforward calculations we get
\begin{lem} \label{lem:tilde symmetry}
	We have
	\begin{equation}
		\tilde \tau P \tilde \tau^* = Q, \quad \tilde \tau Q \tilde \tau^* = P.
	\end{equation}
	It follows that $\tilde \tau B \tilde \tau^* = B$ and $\tilde \tau A \tilde \tau^* = -A$, hence $\tilde \tau E_{\lambda} = E_{-\lambda}$ for all eigenvalues $\lambda$ of $A$.
\end{lem}

\begin{lem} \label{lem:Kramers degeneracy}
	For $\lambda \not\in \{-1, 0, 1\}$, the spaces $E_{\lambda}$ have even dimension.
\end{lem}

\begin{proof}
	From Lemma \ref{lem:spectral symmetry} we have that $B$ is a linear isomorphism from $E_{\lambda}$ to $E_{-\lambda}$ so $(B^* B)^{-1/2}B$ is a unitary from $E_{\lambda}$ to $E_{-\lambda}$. By Lemma \ref{lem:tilde symmetry}, $\tilde \tau$ maps $E_{-\lambda}$ to $E_{\lambda}$ so the anti-unitary $\theta = (B^* B)^{-1/2} B \tilde \tau$ maps $E_{\lambda}$ to itself and is odd:
	\begin{equation}
		\theta^2 = \left( (B^* B)^{-1/2} B \tilde \tau \right)^2 = (B^* B)^{-1} B^2 \tilde \tau^2 = -(B^* B)^{-1} B^* B = -\I
	\end{equation}
	where we used $\tilde \tau^2 = -\I$ and $B \tilde \tau = \tilde \tau B$ (cf. Lemma \ref{lem:tilde symmetry}).
	By Lemma \ref{lem:quaternionic structure} this leads to Kramers degeneracy, \ie $\dim E_{\lambda}$ is even.
\end{proof}

\subsection{Decoupling}

The vanishing of $A = U P U^* - P$ means that the unitary $U$ leaves the subspaces $\Ran P$ and $\Ran P^\perp$ invariant. In this case, we call $U$ ``decoupled'' (with respect to $P$). Vice versa, if $A$ does not vanish, $U$ ``couples'' $\Ran P$ and $\Ran P^\perp$.
The following proposition states that we can always decouple $U$ if the $+1$-eigenspace of $A$ is even-dimensional and almost decouple $U$ if it is odd-dimensional:
\begin{prop} \label{prop:maximal decoupling}
	Let $U$ and $P$ be as above. Then there exists a unitary $W$ with $U - W$ compact, $\tau W \tau^* = W^*$ and:
	\begin{itemize}
		\item If $\dim \ker (A - \I) \,\,\, \mod \, 2 = 0$, then $W P W^* - P = 0$.
		\item If $\dim \ker (A - \I) \,\,\, \mod \, 2 = 1$, then
			\begin{equation}
				W P W^* - P = \Pi_+ - \Pi_-
			\end{equation}
			where $\Pi_+$ and $\Pi_-$ are one-dimensional projections.
	\end{itemize}
	Moreover, if $[U, P]$ is Schatten-$p$, then $U - W$ is also Schatten-$p$.
\end{prop}
The remainder of this subsection is devoted to construct the $W$ in this proposition. Taking $W$ to be of the form $W = V U$ we first note that the decoupling condition $[W, P] = 0$ translates to
\begin{equation}
P V = V Q
\end{equation}
where $Q = U P U^*$. We call such a $V$ a ``decoupler''. The symmetry constraint on $W$ implies that $V$ has to satisfy $\tilde\tau V\tilde\tau^*=V^*$, where $\tilde\tau=U\tau$ as above. Moreover, since we want $U - W$ compact, we must have $V-\I$ compact.

If $\dim \ker (A - \I) \,\,\, \mod \, 2 = 1$ we will not be able to find a $V$ satisfying these requirements. Although we do not prove it here, this is actually impossible. 

We construct $V$ in two steps. First, following \cite{LongVersion} we construct a decoupler on the orthogonal complement of $E = E_{+1} \oplus E_{-1}$. In the second step we try to construct a decoupler on $E$. This turns out to be possible only if $\dim \ker (A - \I) \,\,\, \mod \, 2 = 0$. In the other case we can only decouple $W$ up to a two-dimensional subspace.

\subsubsection{Decoupling on the orthogonal complement of \texorpdfstring{$E_{+1} \oplus E_{-1}$}{E}.} \label{sec:decoupling on E^perp}

Define the operator
\begin{equation}
	X = B (\I - 2Q)  = (\I - 2 P) B = \I - P - Q + 2 P Q.
\end{equation}
This operator satisfies the decoupling condition, i.e.
\begin{equation} \label{eq:X intertwines}
	P X = X Q = P Q,
\end{equation}
Moreover, $X - \I = PA  - AQ$ is compact because $A$ is, and if $[U, P]$ is Schatten-$p$, then so is $A = [U, P] U^*$ and therefore also $X$. Yet, $X$ is not unitary:
\begin{equation}
	X X^* = X^* X = \frac{X + X^*}{2} = B^2.
\end{equation}
It follows however from this equation that $X$ is normal and that its spectrum is contained in the circle $\{ z \in \C \, : \, (\Im z)^2 + (\Re z - 1/2)^2 = 1/4 \}$. 

The kernel of $X$ is precisely $E = E_{+1} \oplus E_{-1}$ since it coincides with the kernel of $X^*X = B^2 = A^2 - \I$. 
On its orthogonal complement $E^\perp = E_{(-1, 1)}$ we define a unitary $\Vtw$ by
\begin{equation}
\Vtw := (X^* X)^{-1/2}|_{E^\perp}X|_{E^\perp},
\end{equation}
which is well-defined because $X$ is normal, and $(X^* X)^{-1/2}$ is strictly positive on $E^\perp$. The unitary $\Vtw$ has the same eigenvectors as $X|_{E^\perp}$ but with the eigenvalues rescaled to have modulus one.

Since $X - \I$ is compact, the spectrum of $X$ consists of eigenvalues of finite multiplicity, possibly accumulating at 1. One easily sees that $\Vtw - \I$ is also compact. In fact, if $X - \I$ is Schatten-$p$, then so is $\Vtw - \I$, see Lemma \ref{lem:Schatten-p lemma}. Moreover, since by \eqref{eq:X intertwines} $X$ maps the range of $Q$ into the range of $P$, so does its partial isometry. Since $\Vtw$ is the partial isometry of $X|_{E^\perp}$ it follows that
\begin{equation} \label{eq:widetilde V decouples}
	P|_{E^\perp} \Vtw = \Vtw Q|_{E^\perp}
\end{equation}
where we use that both $P$ and $Q$ leave $E^\perp$ invariant.

Note now that since $\tau U \tau^* = U^*$, the subspace $E^\perp$ is invariant under $\tilde \tau = \tau U$. Indeed, by Lemma \ref{lem:tilde symmetry}, we have $\tilde \tau E_{\lambda} = E_{-\lambda}$ for any $\lambda \in \R$ so in particular, $E^\perp = E_{(-1, 1)}$ is invariant under $\tilde \tau$. Moreover, by Lemma \ref{lem:tilde symmetry} we have $\tilde \tau X \tilde \tau^* = X^*$ which implies
\begin{equation}
	\tilde \tau \Vtw \tilde \tau^* = \Vtw^*
\end{equation}
on $E^\perp$. Thus, $\Vtw$ is a decoupler on $E^\perp$ that is compatible with the symmetry condition of $W$.

\subsubsection{Decoupling on \texorpdfstring{$E_{+1} \oplus E_{-1}$}{E}.} \label{sec:decoupling on E}

It remains to find an as-good-as-possible decoupler on the remaining subspace $E = E_{+1} \oplus E_{-1}$. Since the restriction of $Q$ to $E$ is the projection onto $E_{+1}$ and the restriction of $P$ to $E$ is the projection onto $E_{-1}$, such a decoupler has to swap the spaces $E_{+1}$ and $E_{-1}$. We construct a unitary $v$ on $E$ that achieves this as best as possible. Moreover, we require that $\tilde \tau v \tilde \tau^* = v^*$. This symmetry induced constrant is crucial: without it a $v$ swapping $E_{+1}$ and $E_{-1}$ can always be found because these spaces have the same dimension by Lemma \ref{lem:tilde symmetry}.

By Lemma \ref{lem:tilde symmetry}, $\tilde \tau$ maps $E_{+1}$ to $E_{-1}$ and vica versa. Let $\{\phi_1, \cdots, \phi_{2m+k}\}$ be an orthonormal basis of $E_{+1}$ and take accordingly $\{ \tilde \tau \phi_1, \cdots, \tilde \tau \phi_{2m+k} \}$ as orthonormal basis of $E_{-1}$ with $k=0,1$ depending on whether $\dim E_{+1}$ is even or odd. If $\dim \ker (A - \I) \,\,\, \mod \, 2 = \dim E_{+1} \,\,\, \mod \, 2 = 0$ then $E_{+1}$ and $E_{-1}$ are both even dimensional and we take $F_{\rest} = \{ 0 \}$, $F_{+1} = E_{+1}$ and $F_{-1} = E_{-1}$. If $\dim \ker (A - \I) \,\,\, \mod \, 2 = \dim E_{+1} \,\,\, \mod \, 2 = 1$ then $\dim E_{+1} = \dim E_{-1} = 2m + 1$ are odd and we take $F_{+1} = \Span \{ \phi_1, \cdots, \phi_{2m} \} \subset E_{+1}$, we take $F_{-1} = \tilde \tau F_{+1} = \Span\{ \tilde \tau \phi_1, \cdots \tilde \tau \phi_{2m} \}$, and $F_{\rest} = \Span\{ \phi_{2m+1}, \tilde \tau \phi_{2m+1} \}$. In either case, $\tilde \tau$ leaves $F_{+1} \oplus F_{-1}$ invariant and in the chosen basis takes the form
\begin{equation}
	\tilde \tau|_{F_{+1} \oplus F_{-1}} = \begin{bmatrix} 0 && -\I \\ \I && 0 \end{bmatrix} K
\end{equation}
where $K$ is complex conjugation. According to the decomposition $E = F_{\rest} \oplus (F_{+1} \oplus F_{-1})$ we then take $v$ to be
\begin{equation}\label{eq:v}
	v = \I \oplus \begin{bmatrix} 0 && \begin{matrix} 0 && -\I \\ \I && 0 \end{matrix} \\ \begin{matrix} 0 && -\I \\ \I && 0  \end{matrix} && 0  \end{bmatrix}.
\end{equation}
This leaves $F_{\rest}$ invariant, swaps $F_{+1}$ and $F_{-1}$, and satisfies $\tilde \tau v \tilde \tau^* = v^*$.

\subsubsection{Proof of Proposition \ref{prop:maximal decoupling}.}

	Let
	\begin{equation}
		V = v \oplus \Vtw
	\end{equation}
	with $v$ the (almost) decoupler on $E = E_{+1} \oplus E_{-1}$ from \eqref{eq:v} and $\Vtw$ the decoupler on $E^{\perp}$ constructed in Section \ref{sec:decoupling on E^perp}.	By construction, $\tilde \tau V \tilde \tau^* = V^*$ so that $W = V U$ satisfies $\tau W \tau^* = W^*$. Moreover, since $\Vtw - \I$ is compact and $E$ is finite dimensional, $V - \I$ is compact and so is $W - U = (V - \I) U$.	By \eqref{eq:widetilde V decouples} we have that
	\begin{equation}
		(W P W^* - P)|_{E^\perp} = \Vtw Q|_{E^\perp} \Vtw^* - P|_{E^\perp} = 0
	\end{equation}
	where $Q = U P U^*$. 

	It remains to see how $W^* P W - P$ acts on the subspace $E$. We have
	\begin{equation}
		(W P W^* - P)|_E = v Q|_{E} v^* - P_E = v q v^* - p
	\end{equation}
	where $q = Q|_E$ and $p = P|_{E}$ are the projections on $E_{+1}$ and $E_{-1}$, respectively.
	If $\dim \ker (A - \I) \,\,\, \mod \, 2 = 0$, the unitary $v$ swaps $E_{+1}$ and $E_{-1}$ so $v q v^* - p = 0$ and therefore $W P W^* - P = 0$. 
	If $\dim \ker (A - \I) \,\,\, \mod \, 2 = 1$, we decompose $E=F_\rest\oplus(F_{+1}\oplus F_{-1})$ as above. The unitary $v$ in \eqref{eq:v} leaves $F_\rest$ invariant and swaps $F_{+1}$ and $F_{-1}$. Thus $(v q v^* - p)|_{F_{+1}\oplus F_{-1}} = 0$ and $(v q v^* - p)|_{F_\rest} = \Pi_{+} - \Pi_-$ 
	where $\Pi_+$ and $\Pi_-$ are one-dimensional projections. 
	We conclude that
	\begin{equation}
		W P W^* - P = \Pi_+ - \Pi_-.
	\end{equation}

	Finally, we saw in Section \ref{sec:decoupling on E^perp} that if $[U, P]$ is Schatten-$p$ then $\Vtw - \I$ is Schatten-$p$. Since $E$ is finite dimensional, $V - \I$ is also Schatten-$p$. Therefore $W - U = (V - \I) U$ is Schatten-$p$ and so is
	\begin{equation}
		[W, P] = [(W - U), P] + [U, P].
	\end{equation}
	\hfill\qedsymbol

\subsection{The symmetric Wold construction}

For two projections $P$ and $P'$ we write $P \preceq P'$ if $\Ran P \subset \Ran P'$. In the proof of the symmetric Wold decomposition we will need the following abstract result:
\begin{prop} \label{prop:TR Wold decomposition}
	Let $W$ be a unitary and $P$ a projection with $\tau W \tau^* = W^*$ and $\tau P \tau^* = P$ for an odd time-reversal symmetry $\tau$ and such that
	\begin{equation}
		W P W^* - P = \Pi_+ - \Pi_-
	\end{equation}
	where $\Pi_+$ and $\Pi_-$ are one-dimensional projections. For $k \in \Z$, let $\Pi^{(k)}_+ := \Ad_W^{k-1}(\Pi_+)$ and $\Pi^{(k)}_- := \Ad_{W^*}^k(\Pi_-)$. Then:
	\begin{enumerate}
		\item These projections are mutually orthogonal, \ie
		\begin{equation}\label{eq:orth_claims}
			\Pi_{\sigma}^{(k)} \Pi_{\sigma'}^{(l)} = \delta_{k,l} \delta_{\sigma,\sigma'} \Pi_{\sigma}^{(k)}
		\end{equation}
		for all $k, l \in \Z$ and $\sigma, \sigma' \in \{+, -\}$. 
		\item For $k \leq 0$ we have $\Pi_{+}^{(k)}, \Pi^{(k)}_- \preceq P$ while for $k \geq 1$ we have $\Pi_{+}^{(k)}, \Pi_-^{(k)} \preceq P^\perp$. 
		\item For all $k \in \Z$ the projections $\Pi_+^{(k)}$ and $\Pi_-^{(k)}$ form a Kramers pair, \ie
		\begin{equation}\label{eq:Kramers pairs in Wold decomposition}
			\tau \Pi_{\pm}^{(k)} \tau^* = \Pi_{\mp}^{(k)}.
		\end{equation}
	\end{enumerate}
\end{prop}

\begin{proof}
	We first prove the last statement of the proposition. Write $A = \Ad_W(P) - P = \Pi_+^{(1)} - \Pi_-^{(0)}$. From the assumptions we have
	\begin{equation}
	\begin{aligned}
		\tau \Pi_+^{(1)} \tau^* - \tau \Pi_-^{(0)} \tau^* &= \tau A \tau^* =  W^* P W - P = - \Ad_{W^*}(A) = \Ad_{W^*}(\Pi_-^{(0)}) - \Ad_{W^*}(\Pi_+^{(1)}) \\
								  &= \Pi_-^{(1)} - \Pi_+^{(0)}
	\end{aligned}
	\end{equation}
	so $\tau \Pi_+^{(1)} \tau^* = \Pi_-^{(1)}$ and $\tau \Pi_-^{(0)}\tau^*=\Pi_+^{(0)}$ and the claim follows for $k = 0, 1$. For any $k \in \Z$, \eqref{eq:Kramers pairs in Wold decomposition} follows from
	\begin{equation}
		\tau \Pi_+^{(k)} \tau^* = \tau \Ad_W^k(\Pi_+^{(0)}) \tau^* = \Ad_{W^*}^k(\Pi_-^{(0)}) = \Pi_-^{(k)}.
	\end{equation}	
	
	Since all these projections are one-dimensional, it further follows by Kramers degeneracy that
	\begin{equation} \label{eq:Kramers orthogonality of Pies}
		\Pi_+^{(k)} \Pi_-^{(k)} = 0 
	\end{equation}
	for all $k \in \Z$, which proves part of the orthogonality claims. We prove the remaining orthogonality claims in \eqref{eq:orth_claims} and the second statement by induction. 
	The $N^{\text{th}}$ induction hypothesis is that the family of one-dimensional projections $\{ \Pi_+^{(k)}, \Pi_-^{(k)} \}_{k = -N+1}^N$ is mutually orthogonal, and $\Pi_{\pm}^{(k)} \preceq P$ for $k = -N-1, \cdots, 0$ while $\Pi_{\pm}^{(k)} \preceq P^\perp$ for $k = 1, \cdots, N$.

	\textbf{Base case $N = 1 $:} We must show that $\{\Pi_{\pm}^{(0)}, \Pi_{\pm}^{(1)}\}$ form an orthogonal set, and $\Pi_{\pm}^{(0)} \subset P$ while $\Pi_{\pm}^{(1)} \subset P^\perp$. Since $A=\Pi_+^{(1)} - \Pi_-^{(0)}$ is self adjoint, we get $\Pi_+^{(1)} \perp \Pi_-^{(0)}$. Since conjugation by $\tau$ preserves orthogonality we get from  \eqref{eq:Kramers pairs in Wold decomposition} that also $\Pi_-^{(1)} \perp \Pi_+^{(0)}$.  Moreover, $\Pi_{+}^{(0)} \perp \Pi_-^{(0)}$ and $\Pi_{+}^{(1)} \perp \Pi_-^{(1)}$ by \eqref{eq:Kramers orthogonality of Pies}. Finally, since
	\begin{equation}
		A = W P W^* - P = \Pi_+^{(1)} - \Pi_-^{(0)}
	\end{equation}
	we have $\Pi_-^{(0)} \preceq P$ and $\Pi_+^{(1)} \preceq  P^\perp$. Moreover, since $P$ is $\tau$-invariant also $\Pi_+^{(0)} \preceq  P$ and $\Pi_-^{(1)} \preceq  P^\perp$, which proves the claims about inclusions in $P$ and $P^\perp$, and also the remaining orthogonality claims $\Pi_+^{(0)} \perp \Pi_+^{(1)}$ and $\Pi_-^{(0)} \perp \Pi_-^{(1)}$.

	\textbf{Induction step:} We assume the $N^{\text{th}}$ induction hypothesis to hold and derive the $(N+1)^{st}$. Let
	\begin{equation}
		P^{(N)} = P + \sum_{n = 1}^N \left( \Pi_+^{(n)} + \Pi_-^{(n)} \right).
	\end{equation}
	By the $N^{th}$ induction hypothesis and   \eqref{eq:Kramers pairs in Wold decomposition} this is a $\tau$-invariant projection and $P \preceq P^{(N)}$. Consider 
	\begin{equation}
		A^{(N)} := W P^{(N)} W^* - P^{(N)} = \Pi_+^{(N+1)} - \Pi_-^{(N)}.
	\end{equation}
	Since this is a difference of projections it follows that $\Pi_+^{(N+1)} \preceq (P^{(N)})^{\perp} \preceq P^\perp$, and by $\tau$-invariance also $\Pi_-^{(N+1)} \preceq P^\perp$. The claims $\Pi_{\pm}^{(-N)} \preceq P$ follow similarly by setting $P^{(-N)} = P - \sum_{n = 0}^{N-1} \left( \Pi_+^{(-n)} + \Pi_-^{(-n)} \right)$ and considering
	\begin{equation}
		A^{(-N)} := W P^{(-N)} W^* - P^{(-N)} = \Pi_+^{(-N+1)} - \Pi_-^{(-N)}.
	\end{equation}
	This proves the claims about inclusions in $P$ and $P^\perp$.

	To prove the orthogonality claims, note that $(\tau W^n)^2 = -\I$ for any $n$, and since
	\begin{equation}
		(\tau W^n) \Pi_+^{(m - n)} (\tau W^n)^* = \tau^* \Pi_+^{(m)} \tau = \Pi_-^{(m)}
	\end{equation}
	for any $m$, we see that $\Pi_+^{n} \perp \Pi_-^{(m)}$ for all $n, m \in \Z$.

	It remains to show that $\{\Pi_+^{(n)}\}_{n = -N}^{N+1}$ and $\{\Pi_-^{(n)}\}_{n = -N}^{N+1}$ are othogonal families. We already know that $\Pi_{\pm}^{N+1}$ is orthogonal to $\Pi_{\pm}^{n}$ for $n = -N, \cdots, 0$ because the former is a subprojection of $P^{\perp}$ while the latter are subprojections of $P$. Similarly, we know that $\Pi_{\pm}^{(-N)}$ is orthogonal to $\Pi_{\pm}^{(n)}$ for $n = 1, \cdots, N+1$.

	To see that $\Pi_{\pm}^{(N+1)}$ is orthogonal to $\Pi_{\pm}^{(n)}$ for $n = 1, \cdots, N$ we simply note that by induction hypothesis for any such $n$ we have $\Pi_{\pm}^{(0)} \perp \Pi_{\pm}^{(N+1-n)}$. Since conjugation by $W$ preserves orthogonality, we conclude that $\Pi_{\pm}^{(n)} \perp \Pi_{\pm}^{(N+1)}$ as required. The orthogonality of $\Pi_{\pm}^{(-N)}$ with $\Pi_{\pm}^{(n)}$ for $n = -N+1, \cdots, 0$ is proved in the same way. This proves all the required orthogonality relations and thereby concludes the induction step.
\end{proof}

\subsection{Proof of Theorem \ref{thm:symmetric Wold}}

	If $\dim \ker (A - \I) \,\,\, \mod \, 2 = 0$, all claims follow directly from Propsition \ref{prop:maximal decoupling}. 
	
	If $\dim \ker (A - \I) \,\,\, \mod \, 2 = 1$, then Proposition \ref{prop:maximal decoupling} provides a unitary $W$ with $U - W$ compact, $\tau W \tau^* = W^*$ and such that
	\begin{equation}
		\Ad_W(P) - P = \Pi_+ - \Pi_-
	\end{equation}
	with one-dimensional projections $\Pi_+$ and $\Pi_-$. Moreover, if $[U, P]$ is Schatten-$p$, then so is $U - W$. 
	
	For $k \in \Z$, let $\Pi^{(k)}_+ := \Ad_W^{k-1}(\Pi_+)$ and $\Pi^{(k)}_- := \Ad_{W^*}^k(\Pi_-)$. Decompose the Hilbert space as $\caH = \caH' \oplus \caH''$ where $\caH' = \bigoplus_{k \in \Z, \sigma \in \{+, -\}} \Ran \Pi_{\sigma}^{(k)}$. By Proposition \ref{prop:TR Wold decomposition}, both $W$ and $P$ leave $\caH'$ and $\caH''$ invariant so $W = W' \oplus W''$ and $P = P' \oplus P''$. Since $\Ran \Pi_+, \Ran \Pi_- \subset \caH'$, we have $[W'', P''] = 0$, and it only remains to identify $W'$ with the unitary $S$ described at the end of Section \ref{sec:symmetric Wold statement}.

	Let $\phi_1$ be a unit vector spanning $\Ran \Pi_+^{(1)}$ and define $\phi_k = W^{k-1} \phi_1$ for all $k \in \Z$. Then $\phi_k$ spans $\Ran \Pi_+^{(k)}$. For each $k \in \Z$, let $\overline{\phi}_k = \tau \phi_k$. Then $\overline{\phi}_k$ spans $\Ran \Pi_-^{(k)}$ and $\caH' = \Span \{ \phi_k, \overline{\phi}_k \, : \, k \in \Z\}$. On this space, the unitary $W'$ acts as
	\begin{equation}
	W' \phi_k = \phi_{k+1}, \quad  W' \overline{\phi}_k = \overline{\phi}_{k-1}
	\end{equation}
	for all $k \in \Z$. Indeed, $\phi_{k+1} = W^{k} \phi_1 = W \phi_k$ while $\overline{\phi}_{k} = \tau \phi_{k} = \tau W^{k-1} \phi_1 = W^{1-k} \tau \phi_1 = W^{1-k} \overline{\phi}_1$, so $\overline{\phi}_{k-1} = W^{2-k} \overline{\phi}_1 = W \overline{\phi}_k$. The unitary $W'$ is equivalent to $S$ in \eqref{eq:shiftau} by the isomorphism $\phi_x \mapsto \ket{x,+}$ and $\overline{\phi}_x \mapsto \ket{x,-}$.
	This concludes the proof.
	\hfill\qedsymbol

\appendix

\section{Kramers Degeneracy}
\begin{lem} \label{lem:quaternionic structure}
	Let $\theta$ be an anti-unitary operator on a finite dimensional Hilbert space $V$ with $\theta^2 = -\I$. Then $V$ is even-dimensional and has an orthonormal basis consisting of Kramers pairs, \ie there is an orthonormal basis $\{ \phi_1, \phi'_1, \cdots \phi_n, \phi'_n \}$ such that $\theta \phi_i = \phi'_i$ and $\theta \phi'_i = -\phi_i$ for all $i = 1, \cdots, n$.
\end{lem}

\begin{proof}
	
	Let $\phi_1$ be any vector in $V$ of unit length, and put $\phi'_1 = \theta \phi_1$. Then
	\begin{equation}
		\langle \phi_1, \phi'_1 \rangle = \langle \phi_1, \theta \phi_1 \rangle = \overline{ \langle \theta \phi_1, \theta^2 \phi_1 \rangle} = - \langle \phi_1, \theta \phi_1 \rangle = - \langle \phi_1, \phi'_1 \rangle,
	\end{equation}
	\ie $\phi_1 \perp \phi'_1$ and $\theta \phi'_1 = \theta^2 \phi_1 = -\phi_1$. Now pick any vector $\phi_2$ in the orthongonal complement of $\Span\{ \phi_1, \phi'_1\}$ and put $\phi'_2 = \theta \phi_2$. By the same reasoning as before, $\phi_2 \perp \phi'_2$ and $\theta \phi'_2 = -\phi_2$. Repeating this construction eventually yields the required basis $\{\phi_1, \phi'_1, \cdots, \phi_1, \phi'_n\}$.
\end{proof}

\section{Schatten-\texorpdfstring{$p$}{p} Lemma}

\begin{lem} \label{lem:Schatten-p lemma}
	Let $X$ be a normal operator with $X - \I$ compact, spectrum contained in the circle $\{ z \in \C \, : \, (\Im z)^2 + (\Re z - 1/2)^2 = 1/2\}$ and empty kernel. Let
	\begin{equation}
		V := (X^* X)^{-1/2}X.
	\end{equation}
	Then if $X - \I$ is Schatten-$p$, so is $V - \I$.
\end{lem}

\begin{proof}
	Let $\{\lambda_i\}_{i \in \N}$ be the non-zero eigenvalues of the compact operator $X - \I$ ordered such that $\abs{\lambda_i} \geq \abs{\lambda_{i+1}}$ for all $i$. This operator is Schatten-$p$ if and only if
	\begin{equation}
		\sum_{i \in \N} \abs{\lambda_i}^p < \infty.
	\end{equation}
	The non-zero eigenvalues of $V - \I$ are
	\begin{equation}
		\mu_i = \frac{\lambda_i + 1}{\abs{\lambda_i + 1}} -1.
	\end{equation}
	Note that since $X - \I$ has empty kernel, none of the $\lambda_i$ equal $-1$ so the $\mu_i$ are always well defined.

	For $\mu_i$ close to 0 we have $\abs{\mu_i} = \abs{\lambda_i} + \caO(\abs{\lambda_i}^2)$ so
	\begin{equation}
		\abs{\mu_i}^s = \abs{\lambda_i}^s + s \caO( \abs{\lambda_i}^2 \abs{\lambda_i}^{s-1}  ) = \abs{\lambda_i}^s + \caO( \abs{\lambda_i}^{s+1}  ).
	\end{equation}
	Both terms on the right-hand side are summable by assumtion, so
	\begin{equation}
		\sum_i \abs{\mu_i}^s < \infty,
	\end{equation}
	\ie $V - \I$ is Schatten-$p$.
\end{proof}

\section*{Acknowledgements}
A. Bols thanks the Villum Fonden through the QMATH Centre of Excellence (grant no. 10059). C. Cedzich acknowledges support by the Deutsche Forschungsgemeinschaft (DFG, German Research Foundation) -- project number 441423094.

\bibliographystyle{abbrvArXiv}
\bibliography{bib}

\end{document}